\documentclass[a4paper,UKenglish,cleveref, autoref, thm-restate]{lipics-v2021}

\usepackage{color}
\usepackage{booktabs}
\usepackage{xspace}

\newcommand{\N}{\mathbb{N}}
\newcommand{\Z}{\mathbb{Z}}

\newcommand{\Geqt}{G_{\Delta}}
\newcommand{\pasc}{PASC algorithm\xspace}

\DeclareMathOperator{\argmax}{argmax}
\DeclareMathOperator{\argmin}{argmin}
\DeclareMathOperator{\id}{id}
\DeclareMathOperator{\succr}{succ}
\DeclareMathOperator{\Prob}{Pr}

\pdfoutput=1 
\hideLIPIcs  

\title{The Structural Power of Reconfigurable Circuits in the Amoebot Model}

\author{Andreas Padalkin}{Paderborn University, Germany}{andreas.padalkin@upb.de}{https://orcid.org/0000-0002-4601-9597}{}
\author{Christian Scheideler}{Paderborn University, Germany}{scheideler@upb.de}{https://orcid.org/0000-0002-5278-528X}{}
\author{Daniel Warner}{Paderborn University, Germany}{dwarner@upb.de}{https://orcid.org/0000-0002-9423-6094}{}

\authorrunning{A. Padalkin, C. Scheideler and D. Warner} 

\Copyright{Andreas Padalkin, Christian Scheideler and Daniel Warner} 



\category{} 

\relatedversion{} 

\acknowledgements{This work has been supported by the DFG Project SCHE 1592/6-1 (PROGMATTER).}

\nolinenumbers 

\EventEditors{John Q. Open and Joan R. Access}
\EventNoEds{2}
\EventLongTitle{42nd Conference on Very Important Topics (CVIT 2016)}
\EventShortTitle{CVIT 2016}
\EventAcronym{CVIT}
\EventYear{2016}
\EventDate{December 24--27, 2016}
\EventLocation{Little Whinging, United Kingdom}
\EventLogo{}
\SeriesVolume{42}
\ArticleNo{23}

\begin{document}

\maketitle

\begin{abstract}


    The \emph{amoebot model} [Derakhshandeh et al., 2014] has been proposed as a model for programmable matter consisting of tiny, robotic elements called \emph{amoebots}.
    We consider the \emph{reconfigurable circuit extension} [Feldmann et al., JCB 2022] of the geometric (variant of the) amoebot model that allows the amoebot structure to interconnect amoebots by so-called \emph{circuits}.
    A circuit permits the instantaneous transmission of signals between the connected amoebots.
    In this paper, we examine the structural power of the reconfigurable circuits.

    We start with some fundamental problems like the \emph{stripe computation problem} where, given any connected amoebot structure $S$, an amoebot $u$ in $S$, and some axis $X$, all amoebots belonging to axis $X$ through $u$ have to be identified.
    Second, we consider the \emph{global maximum problem}, which identifies an amoebot at the highest possible position with respect to some direction in some given amoebot (sub)structure.
    A solution to this problem can then be used to solve the \emph{skeleton problem}, where a (not necessarily simple) cycle of amoebots has to be found in the given amoebot structure which contains all boundary amoebots.
    A canonical solution to that problem can then be used to come up with a canonical path, which provides a unique characterization of the shape of the given amoebot structure.
    Constructing canonical paths for different directions will then allow the amoebots to set up a spanning tree and to check symmetry properties of the given amoebot structure.

    The problems are important for a number of applications like rapid shape transformation, energy dissemination, and structural monitoring.
    Interestingly, the reconfigurable circuit extension allows polylogarithmic-time solutions to all of these problems.
\end{abstract}

\section{Introduction}
\label{sec:intro}

The \emph{amoebot model} \cite{DBLP:conf/spaa/DerakhshandehDGRSS14,DBLP:conf/wdag/DaymudeRS21} is a well-studied model for programmable matter \cite{DBLP:journals/ijhsc/ToffoliM93} -- a substance that can be programmed to change its physical properties, like its shape and density.
In the geometric variant of this model, the substance (called the \emph{amoebot structure}) consists of simple particles (called \emph{amoebots}) that are placed on the infinite triangular grid graph and are capable of local movements through \emph{expansions} and \emph{contractions}.

Inspired by the \emph{nervous} and \emph{muscular system}, Feldmann et al.~\cite{FPSD21} introduced a \emph{reconfigurable circuit extension} to the amoebot model with the goal of significantly accelerating fundamental problems like leader election and shape transformation.
As a first step, they showed that leader election, consensus, compass alignment, chirality agreement, and various shape recognition problems can be solved in at most $O(\log n)$ time.
This paper continues this line of work by considering a number of additional problems:

First, we consider the \emph{stripe computation problem} where, given any connected amoebot structure $S$, an amoebot $u$ in $S$, and some axis $X$, all amoebots belonging to axis $X$ through $u$ have to be identified.
Second, we consider the \emph{global maximum problem}, which identifies an amoebot at the highest possible position with respect to some direction in some given amoebot (sub)structure.
A solution to this problem can then be used to solve the \emph{skeleton problem}, where a (not necessarily simple) cycle of amoebots has to be found in the given amoebot structure which contains all boundary amoebots.
A canonical solution to that problem can then be used to come up with a canonical path, which provides a unique characterization of the shape of the given amoebot structure.
Constructing canonical paths for different directions will then allow the amoebots to set up a spanning tree and to check symmetry properties of the given amoebot structure.

The problems have a number of important applications.
The stripe computation problem is important to avoid conflicts in joint amoebot contractions and expansions (see \cref{fig:joint_movement}), which is critical for rapid shape transformation.
A spanning tree is an important step towards energy distribution from amoebots with access to energy to amoebots without such access \cite{DBLP:conf/icdcn/DaymudeRW21}, and canonical skeleton paths as well as symmetry checks are important for structural monitoring and repair.

\begin{figure}[tb]
    \centering
    \begin{subfigure}[b]{0.2\textwidth}
        \centering
        \includegraphics[page=1]{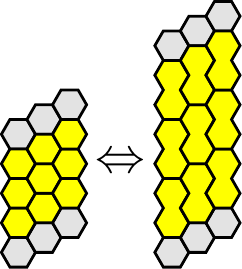}
        \caption{}
        \label{fig:joint_movement:example}
    \end{subfigure}
    \hfill
    \begin{subfigure}[b]{0.29\textwidth}
        \centering
        \includegraphics[page=1]{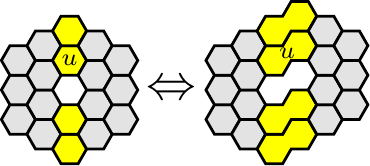}
        \caption{}
        \label{fig:joint_movement:stripe}
    \end{subfigure}
    \hfill
    \begin{subfigure}[b]{0.13\textwidth}
        \centering
        \includegraphics[page=1]{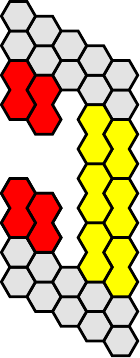}
        \caption{}
        \label{fig:conflicts1}
    \end{subfigure}
    \hfill
    \begin{subfigure}[b]{0.13\textwidth}
        \centering
        \includegraphics[page=1]{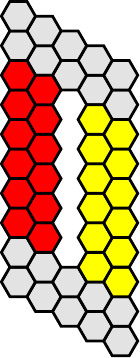}
        \caption{}
        \label{fig:conflicts2}
    \end{subfigure}
    \caption{
        Feldmann et al.~\cite{FPSD21} have proposed joint movements to the amoebot model where
        an expanding amoebot is capable of pushing other amoebots away from it,
        and a contracting amoebot is capable of pulling other amoebots towards it.
        (a) and (b) show a joint expansion (from left to right) resp. a joint contraction (from right to left) of the yellow amoebots.
        The left side of (b) shows stripe $\operatorname A(S, u, N)$ (see Section~\ref{sec:problem}).
        The stripe can expand without causing any conflicts.
        (c) and (d) show exemplary conflicts.
        (c) If the yellow amoebots contract, the red amoebots will collide.
        In order to avoid the collision, the red amoebots have to contract as well.
        (d) If the yellow amoebots expand, the amoebot structure will tear apart.
        In order to maintain all connections, the red amoebots have to expand as well.
    }
    \label{fig:joint_movement}
    \label{fig:conflicts}
\end{figure}

\subsection{Geometric Amoebot Model}
\label{sec:model:amoebot}

In the \emph{geometric amoebot model} \cite{DBLP:conf/wdag/DaymudeRS21}, a set of $n$ amoebots is placed on the infinite regular triangular grid graph $\Geqt = (V, E)$ (see Figure~\ref{fig:model_classic}).
An amoebot is an anonymous, randomized finite state machine that either occupies one or two adjacent nodes of $\Geqt$, and every node of $\Geqt$ is occupied by at most one amoebot.
If an amoebot occupies just one node, it is called \emph{contracted} and otherwise \emph{expanded}.
Two amoebots that occupy adjacent nodes in $\Geqt$ are called \emph{neighbors}.
%
Amoebots are able to move through \emph{contractions} and \emph{expansions}.
However, since our algorithms do not make use of movements, we omit further details and refer to \cite{DBLP:conf/wdag/DaymudeRS21} for more information.

Each amoebot has a compass orientation (it defines one of its incident edges as the northern direction) and a chirality (a sense of clockwise or counterclockwise rotation) that it can maintain as it moves, but initially the amoebots might not agree on their compass orientation and chirality.
In this paper, we assume that all amoebots share a common compass orientation and chirality.
This is reasonable since Feldmann et al.~\cite{FPSD21} showed that all amoebots are able to come to an agreement within the considered extension (see \cref{sec:related}).

Let the \emph{amoebot structure}~$S \subseteq V$ be the set of nodes occupied by the amoebots.
By abuse of notation, we identify amoebots with their nodes.
We say that $S$ is \emph{connected} iff $G_S$ is connected, where $G_S = \Geqt|_S$ is the graph induced by~$S$.
In this paper, we assume that initially, $S$ is connected and all amoebots are contracted.
Also, we assume the fully synchronous activation model, i.e., the time is divided into synchronous rounds, and every amoebot is active in each round.
On activation, each amoebot may perform a movement and update its state as a function of its previous state.
However, if an amoebot fails to perform its movement, it remains in its previous state.
The time complexity of an algorithm is measured by the number of synchronized rounds required by it.

\begin{figure}[tb]
    \centering
    \begin{subfigure}[b]{0.2\textwidth}
        \centering
        \includegraphics[page=1]{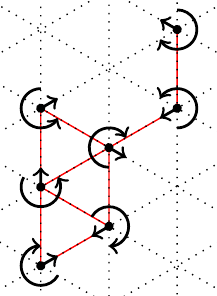}
        \caption{}
        \label{fig:model_classic}
    \end{subfigure}
    \hfill
    \begin{subfigure}[b]{0.2\textwidth}
        \centering
        \includegraphics[page=1]{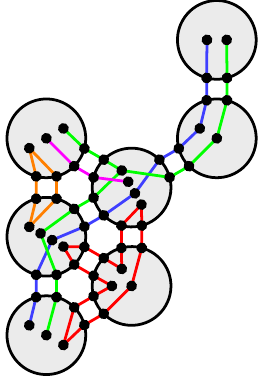}
        \caption{}
        \label{fig:model_graph}
    \end{subfigure}
    \hfill
    \begin{subfigure}[b]{0.2\textwidth}
        \centering
        \includegraphics[page=1]{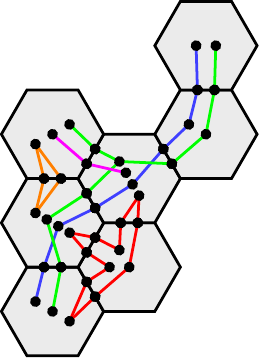}
        \caption{}
        \label{fig:model_hex}
    \end{subfigure}
    \hfill
    \begin{subfigure}[b]{0.25\textwidth}
        \centering
        \includegraphics[page=1]{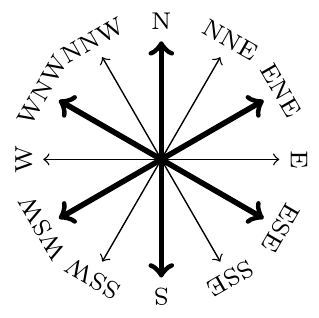}
        \caption{}
        \label{fig:compass}
    \end{subfigure}
    \caption{
        (a) shows an amoebot structure $S$.
        The dotted lines indicate the triangular grid $\Geqt$.
        The nodes indicate the amoebots. The arrows show their chirality and compass orientation.
        The red edges indicate the graph $G_S$.
        (b) and (c) show an amoebot structure with $k = 2$ external links between neighboring amoebots.
        The amoebots are shown in gray.
        The nodes on the boundary are the pins,
        and the ones within the amoebots the partition sets.
        An edge between a partition set $Q$ and a pin $p$ implies $p \in Q$.
        Each color indicates another circuit.
        (a) and (b) are taken from \cite{FPSD21}.
        (d) shows the cardinal directions.
        The thick arrows indicate the cardinal directions along the main axes,
        and the thin ones the cardinal directions perpendicular to the main axes.
    }
    \label{fig:model}
\end{figure}

\subsection{Reconfigurable Circuit Extension}
\label{sec:model:circuits}

In the \emph{reconfigurable circuit extension} \cite{FPSD21},
each edge between two neighboring amoebots $u$ and $v$ is replaced by $k$ edges called \emph{external links} with endpoints called \emph{pins}, for some constant $k \ge 1$ that is the same for all amoebots.
For each of these links, one pin is owned by $u$ while the other pin is owned by $v$.
In this paper, we assume that neighboring amoebots have a common labeling of their incident external links.

Each amoebot $u$ \emph{partitions} its \emph{pin set} $P(u)$ into a collection $\mathcal Q(u)$ of pairwise disjoint subsets such that the union equals the pin set, i.e., $P(u) = \bigcup_{Q \in \mathcal Q(u)} Q$.
We call $\mathcal Q(u)$ the \emph{pin configuration} of $u$ and $Q \in \mathcal Q(u)$ a \emph{partition set} of $u$.
Let $\mathcal Q = \bigcup_{u \in S} \mathcal Q(u)$ be the collection of all partition sets in the system.
Two partition sets are \emph{connected} iff there is at least one external link between those sets.
Let $L$ be the set of all connections between the partition sets in the system.
Then, we call $H=(\mathcal Q,L)$ the \emph{pin configuration} of the system and any connected component $C$ of $H$ a \emph{circuit} (see Figure~\ref{fig:model_graph}).
Note that if each partition set of $\mathcal Q$ is a \emph{singleton}, i.e., a set with exactly one element, then every circuit of $H$ just connects two neighboring amoebots.
However, an external link between the neighboring amoebots $u$ and $v$ can only be maintained as long as both, $u$ and $v$ occupy the incident nodes.
Whenever two amoebots disconnect, the corresponding external links and their pins are removed from the system.
An amoebot is part of a circuit iff the circuit contains at least one of its partition sets.
A priori, an amoebot $u$ may not know whether two of its partition sets belong to the same circuit or not since initially it only knows $\mathcal Q(u)$.

Each amoebot $u$ can send a primitive signal (a \emph{beep}) via any of its partition sets $Q \in \mathcal Q(u)$ that is received by all partition sets of the circuit containing $Q$ at the beginning of the next round.
The amoebots are able to distinguish between beeps arriving at different partition sets.
More specifically, an amoebot receives a beep at partition set $Q$ if at least one amoebot sends a beep on the circuit belonging to $Q$, but the amoebots neither know the origin of the signal nor the number of origins.
Note that beeps are enough to send whole messages over time, especially between adjacent amoebots.
We modify an activation of an amoebot as follows.
As a function of its previous state and the beeps received in the previous round, each amoebot may perform a movement, update its state, reconfigure its pin configuration, and activate an arbitrary number of its partition sets.
The beeps are propagated on the updated pin configurations.
If an amoebot fails to perform its movement, it remains in its previous state and pin configuration, and does not beep on any of its partition sets.

In this paper, we will utilize the dual graph of the triangular grid graph, i.e., a hexagonal tesselation, to visualize amoebot structures (see Figure~\ref{fig:model_hex}).
Thereby, we reduce each external link to a single pin.
Furthermore, in order to improve the comparability of circuit configurations, we add pins to each side of the hexagon.

\subsection{Problem Statement and Our Contribution}
\label{sec:problem}

Let $D_m = \{ N, ENE, ESE, S, WSW, WNW \}$ be the set of all cardinal directions along the axes of $\Geqt$,
and $D_p = \{ E, SSE, SSW, W, NNW, NNE \}$ the set of all cardinal directions perpendicular to the axes of $\Geqt$ (see Figure~\ref{fig:compass}).
In the following, we state the considered problems.
An overview of our results is given by Table~\ref{tab:results}.

\begin{table}[hbt]
    \caption{
        An overview of our algorithmic results.
    }
    \begin{tabularx}{0.99\textwidth}{XXXll}
        \toprule
        Problem & Required pins & Runtime & Section & Theorem \\
        \midrule
        Stripe & 2 & $O(\log n)$ & Section~\ref{sec:idapp} & \cref{th:stripe} \\
        Global maxima & 2 & $O(\log^2 n)$ w.h.p. & Section~\ref{sec:idapp} & \cref{th:maxima} \\
        Canonical skeleton & 4 & $O(\log^2 n)$ w.h.p. & Section~\ref{sec:skeleton} & \cref{th:skeleton} \\
        Spanning tree & 4 & $O(\log^2 n)$ w.h.p. & Section~\ref{sec:spanningtree} & \cref{th:spanning_tree} \\
        Symmetry detection & 4 & $O(\log^5 n)$ w.h.p. & Section~\ref{sec:symmetry} & \cref{th:symmetry} \\
        \bottomrule
    \end{tabularx}
    \label{tab:results}
\end{table}

First, we consider the \emph{stripe computation problem}.
Let $\operatorname X(v,d)\subseteq V$ denote the nodes of $\Geqt$ that lie on the axis through the node $v \in V$ into the cardinal direction $d \in D_m \cup D_p$.
For $R \subseteq V$, we call the set $\operatorname A(R, v, d) = R \cap \operatorname X(v,d)$ a \emph{stripe} of $R$ (see Figure~\ref{fig:joint_movement:stripe}).
Note that a stripe is not necessarily connected.
Let an amoebot $u \in S$ and a cardinal direction $d \in D_m \cup D_p$ be given, i.e., each amoebot $v \in S$ knows the cardinal direction $d$ and whether $v = u$.
The goal of each amoebot $v \in S$ is to determine whether $v \in \operatorname A(S,u,d)$.
Our \emph{stripe algorithm} solves the stripe computation problem after $O(\log n)$ rounds.

Second, we consider the \emph{global maxima problem}.
Let a cardinal direction $d \in D_m \cup D_p$ and a non-empty set $R \subseteq S$ be given, i.e., each amoebot $v \in S$ knows the direction $d$ and whether $v \in R$.
The goal of each amoebot $v \in S$ is to determine whether $v \in \argmin_{w \in R} \operatorname f_d(R,w)$ where $\operatorname f_d(R,w)$ denotes the number of amoebots in $R$ that lie in direction $d$ from amoebot $w$.
We call $\argmin_{w \in R} \operatorname f_d(R,w)$ the set of global maxima of $R$ with respect to $d$.
Our \emph{global maxima algorithm} solves the global maxima problem after $O(\log^2 n)$ rounds w.h.p.%
\footnote{An event holds \emph{with high probability (w.h.p.)} if it holds with probability at least $1 - 1/n^c$ where the constant $c$ can be made arbitrarily large.}

Third, we consider the \emph{(canonical) skeleton problem}.
An amoebot $u$ is a \emph{boundary amoebot} iff it is adjacent to an unoccupied node in $V \setminus S$.
Otherwise, we call $u$ an \emph{inner amoebot}.
A (potentially non-simple) cycle of amoebots is a skeleton iff the cycle contains all boundary amoebots in S.
Note that the skeleton may contain inner amoebots.
An amoebot structure computes a skeleton $C$ iff each amoebot knows its predecessor and successor for each of its occurrences in $C$.
The goal of the skeleton problem is to compute an arbitrary skeleton.

Since skeletons are not unique, we define a \emph{canonical skeleton} with respect to a cardinal direction $d \in D_m \cup D_p$ and a sign $s \in \{ +, - \}$ (abbreviated as $(d,s)$-skeleton).
We defer the definition to Section~\ref{sec:skeleton}.
Let a cardinal direction $d \in D_m \cup D_p$ and a sign $s$ be given, i.e., each amoebot $v \in S$ knows the cardinal direction $d$ and the sign $s$.
The goal of the canonical skeleton problem is to compute the canonical skeleton.
Our \emph{canonical skeleton algorithm} solves the (canonical) skeleton problem after $O(\log^2 n)$ rounds w.h.p.

Our algorithms for the remaining problems are based on skeletons.
However, they split the skeletons into paths that we call \emph{skeleton paths}.
For the canonical skeletons, we define a \emph{canonical skeleton path} by specifying a splitting point.
Our canonical skeleton algorithm determines this point in parallel to the computation of the canonical skeleton.

Fourth, we consider the \emph{spanning tree problem}.
A \emph{tree} is a cycle-free and connected graph.
A \emph{spanning tree} of an amoebot structure $S$ is a tree $T = (S, E_T)$ with $E_T \subseteq E$. 
An amoebot structure computes a spanning tree $T$ if each amoebot $u \in S$ knows whether $\{ u, v \} \in E_T$ for each neighbor $v \in \operatorname N(u)$.
The goal of the spanning tree problem is to compute an arbitrary spanning tree.
Our \emph{spanning tree algorithm} solves the spanning tree problem after $O(\log^2 n)$ rounds w.h.p.

Finally, we consider the \emph{symmetry detection problem}.
The goal of that problem is to determine whether the amoebot structure features rotational or reflection symmetries.
Our \emph{symmetry detection algorithm} solves the the problem after $O(\log^5 n)$ rounds w.h.p.

\subsection{Related Work}
\label{sec:related}

The reconfigurable circuit extension was introduced by Feldmann et al. \cite{FPSD21}.
They have proposed solutions for \emph{leader election} (see Section~\ref{sec:leader}), \emph{consensus}, \emph{compass alignment}, \emph{chirality agreement}, and various \emph{shape recognition} problems.
Both, the alignment of the compasses and the agreement on a chirality requires $O(\log n)$ w.h.p.
This makes our assumption of a common compass orientation and chirality reasonable.

To our knowledge the stripe computation, global maxima, (canonical) skeleton, and symmetry detection problem have not been considered within the standard amoebot model.
However, regarding the global maxima problem, Daymude et al.~\cite{DBLP:conf/icdcn/DaymudeGHKSR20} have considered the related problem of determining the dimensions of an object (a finite, connected, static set of nodes) in order to solve various \emph{convex hull problems}.
Their approach can be easily adjusted to compute the global maxima of the amoebot structure.
However, it requires $O(n)$ rounds.

The spanning tree problem is widely studied in the distributed algorithms community, e.g., \cite{DBLP:journals/dc/MashreghiK21} (also see the cited papers within the reference).
The \emph{spanning tree primitive} is one of the most used techniques to move amoebots, e.g., \cite{DBLP:series/lncs/DaymudeHRS19,DBLP:conf/icdcn/DaymudeGHKSR20,DBLP:conf/dna/DerakhshandehGS15,DBLP:journals/dc/LunaFSVY20}.
Beyond that, spanning trees were applied to distribute energy \cite{DBLP:conf/icdcn/DaymudeRW21}.
However, the construction requires $\Omega(D)$ rounds where $D$ is the diameter of the amoebot structure (e.g., see \cite{DBLP:series/lncs/DaymudeHRS19}).
Since $D = \Omega(\sqrt n)$, our solution is a significant improvement.


\section{Preliminaries}
\label{sec:prelim}

In this section, we enumerate important primitives given in previous papers.

\subparagraph*{Global Circuit}

If each amoebot partitions its pins into one partition set, we obtain a single circuit that interconnects all amoebots.
We call this circuit the \emph{global circuit} \cite{FPSD21}.

\subparagraph*{Leader Election}
\label{sec:leader}

We make use of a generalized version of the leader election algorithm proposed by Feldmann et al.~\cite{FPSD21}:

\begin{theorem}
\label{th:leaderelection}
    Let $C_1, \dots, C_m$ be sets of candidates.
    For each $i \in \{1, \dots, m\}$, let $\mathcal C_i$ be the circuit that connects all candidates of set $C_i$.
    Let $\mathcal C_i \cap \mathcal C_j = \emptyset$ hold for all $i \neq j$.
    An amoebot structure elects a leader from each set of candidates after $\Theta(\log n)$ rounds w.h.p.
\end{theorem}

In the classical leader election problem, the amoebot structure $S$ has to elect a leader only from the set $C_1 = S$.

\subparagraph*{Chains}

We call an ordered sequence $C = (u_0, \dots, u_{m-1})$ of $m$ amoebots a \emph{chain} iff
(i) all subsequent amoebots $u_i, u_{i+1}$ are neighbors,
(ii) each amoebot in $C$ except $u_1$ knows its predecessor, and
(iii) each amoebot in $C$ except $u_m$ knows its successor.

\subparagraph*{Boundary Sets}
\label{sec:boundary}

We adopt the definition for the \emph{boundary sets} from \cite{DBLP:conf/dna/DerakhshandehGS15}:
Consider $G_{V \setminus S} = \Geqt|_{V \setminus S}$.
The connected components of $G_{V \setminus S}$ are called empty regions.
The number of empty regions is finite since $S$ is finite.
Let $R_1, \dots, R_m$ denote the empty regions.
For $i \in \{ 1, \dots, m\}$, the boundary set $B_i$ is the neighborhood of $R_i$ in $S$, i.e., $B_i = \{ u \in S \mid \exists v \in R_i : \{ u,v \} \in E \}$.
There is exactly one infinite empty region since $S$ is finite.
We call the corresponding boundary set the \emph{outer boundary set},
and the others \emph{inner boundary sets}.
Naturally, each boundary set can be represented as a (potentially non-simple) cycle.


In order to distinguish inner and outer boundaries, we apply the \emph{inner outer boundary test} by Derakhshandeh et al. \cite{DBLP:conf/dna/DerakhshandehGS15}.
They accumulate the angles of the turns while traversing the cycle of the boundary set once.
An outer boundary set results in a sum of $360^\circ$, and an inner boundary set results in a value of $-360^\circ$.
Note that it is sufficient to count the turns by $60^\circ$ modulo $5$.
This allows us to accumulate the sum with constant memory.
However, the traversing requires $O(n)$ rounds.
We accelerate the accumulation by the following result by Feldmann et al.~\cite{FPSD21}:

\begin{theorem}
\label{th:addtree}
    Let $C = (v_0, \dots, v_{m-1})$ be a chain within the amoebot structure where $m \in \N$ denotes the length of the chain.
    Let $k \in \N$ be constant.
    Let $x_i \in \{0, \dots, k-1 \}$ for all $i \in \{ 0, \dots, m-1 \}$.
    Suppose that for each $i \in \{ 0, \dots, m-1 \}$, amoebot $v_i$ knows the value $x_i$.
    Then, the chain computes $x = \sum_{i \in \{ 0, \dots, m-1 \}} x_i \mod k$ after $O(\log m)$ rounds.
\end{theorem}

\begin{corollary}
\label{cor:iobt}
    A boundary set can determine whether it is an inner boundary set or the outer boundary set within $O(\log n)$ rounds w.h.p.
\end{corollary}

\begin{proof}
    In order to accelerate the inner outer boundary test, we apply Theorem~\ref{th:addtree}.
    Recall that each boundary set can organize itself into a cycle.
    We apply Theorem~\ref{th:leaderelection} to split the cycle into a chain.
    Each amoebot knows its predecessor and successor on the cycle and therefore also within the chain.
    Let $x_i$ denote the angle at amoebot $v_i$ and let $k = 5$.
    Finally, note that each boundary has $O(n)$ amoebots since each amoebot has at most three local boundaries.
    Thus, the sum of the angles can be accumulated after $O(\log n)$ rounds.
\end{proof}

\subparagraph*{Synchronization of Procedures}
\label{sec:synchronization}

Sometimes, we want to execute a procedure on different subsets in parallel, e.g., we apply the inner outer boundary test on all boundary sets at once.
The execution of the the same procedure may take different amounts of rounds for each subset.
The amoebots of a single subset are not able to decide when all subsets have terminated.
In order to synchronize the amoebot structure with respect to the procedures,
the amoebots periodically establish the global circuit.
The amoebots of each subset that has not yet terminated beep on that circuit.
The amoebot structure may proceed to the next procedure once the global circuit was not activated.

\section{Computing Identifiers}




In this section, we assign \emph{identifiers} in $\Z$ to the amoebots.
Let $(x_{k-1}, \dots, x_0)$ denote the \emph{two's complement representation} of $- x_{k-1} \cdot 2^{k-1} + \sum_{i = 0}^{k - 2} x_i \cdot 2^i$.
By abuse of notation, we identify the identifiers with their two's complement representation.
Each amoebot computes a two's complement representation of its identifier.
In the first subsection, we compute successive identifiers along chains.
In the second subsection, we compute spatial identifiers with respect to a cardinal direction.
In the third subsection, we show two applications for the identifiers.

\subsection{Successive Identifiers along the Chain}
\label{sec:idchain}
\label{sec:pasc}

In this section, we compute successive identifiers along a chain with respect to an amoebot
that we call the \emph{reference amoebot}.
Let $C = (u_0, \dots, u_{m-1})$ be a chain of amoebots.
Let $u_r$ be an arbitrary reference amoebot within the chain,
e.g., chosen by position (for example $r = 0$), or by a leader election.
We assign identifiers $\id_{C,u_r}$ according to the following rules:
\begin{align}
    \id_{C,u_r}(u_r) &= 0 \label{eq:id_chain_1} \\
    \id_{C,u_r}(u_{i+1}) &= \id_{C,u_r}(u_{i}) + 1 \text{ for } 1 \leq i < m \label{eq:id_chain_2}
\end{align}
Note that $\id_{C,u_r} (u_i) = i - r$.




In order to compute the identifiers,
we utilize a procedure on the chain of amoebots proposed by Feldmann et al.~\cite{FPSD21}
that we henceforth refer to as the \emph{primary and secondary circuit algorithm} (\pasc).
Originally, the algorithm has been used as a subroutine for Theorem~\ref{th:addtree}.

In the following, we explain how the \pasc works (see Figure~\ref{fig:alg:id}).
So, let $C = (u_0, \dots, u_{m-1})$ be a chain of $m$ amoebots.
If an amoebot occurs multiple (but constantly many) times, it operates each position independently of each other
which is possible by using sufficiently many pins.
Each amoebot is either \emph{active} or \emph{passive}.
Initially, each amoebot is active.
The algorithm iteratively transforms active amoebots into passive amoebots while keeping amoebot $u_r$ active.

\begin{figure}[tb]
    \centering
    \includegraphics[page=1]{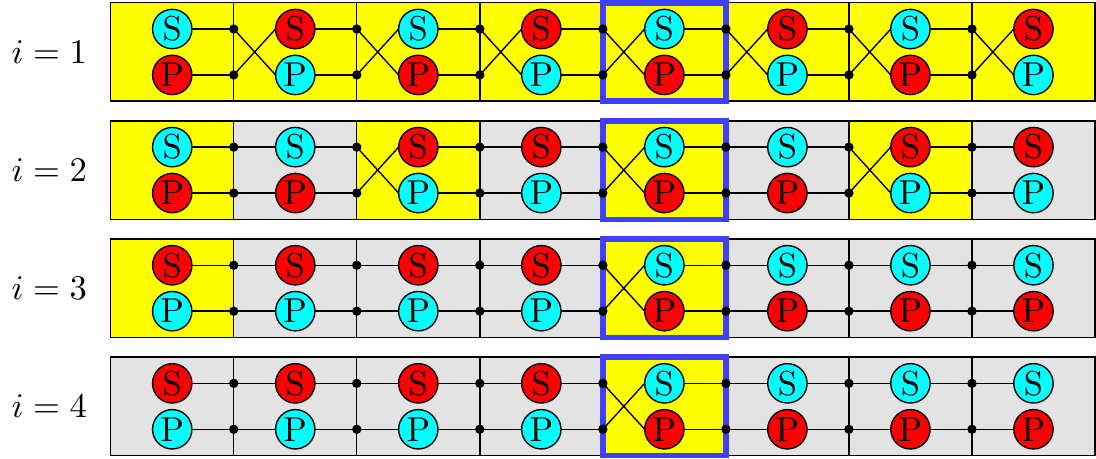}
    \caption{
        \pasc.
        Each figure shows the chain at the beginning of the $i$-th iteration.
        The circles indicate the primary (P) and secondary (S) partition sets.
        The blue bordered amoebot denotes the reference amoebot $u_r$.
        Yellow amoebots are active, and gray amoebots are passive.
        All primary and secondary partition sets that are part of the primary resp.~secondary circuit of $u_r$ are depicted in red resp.~cyan.
    }
    \label{fig:alg:id}
\end{figure}

At the beginning of an iteration, the amoebots establish two circuits as follows.
Each amoebot has two partition sets that we call the \emph{primary and secondary partition set}.
We connect the primary and secondary partition sets by the following two rules%
\footnote{In comparison to \cite{FPSD21}, we have adjusted the rules by mirroring the connections.}
(see Figure~\ref{fig:alg:id}):
(i) If an amoebot is active, we connect its primary partition set to the secondary partition set of its predecessor, and its secondary partition set to the primary partition set of its predecessor.
(ii) If an amoebot is passive, we connect its primary partition set to the primary partition set of its predecessor, and its secondary partition set to the secondary partition set of its predecessor.

We obtain two circuits through all amoebots (see Figure~\ref{fig:alg:id}).
Each amoebot defines the circuit containing its primary partition set as its \emph{primary circuit}, and the circuit containing its secondary partition set as its \emph{secondary circuit}.
Clearly, we obtain two disjoint circuits along the chain.
We refer to \cite{FPSD21} for a detailed construction.

After establishing these circuits, the iteration utilizes two rounds.
In the first round, amoebot $u_r$ activates its primary circuit.
Each active amoebot that has received the beep on its secondary circuit beeps in the second round on its secondary circuit.
These amoebots become passive amoebots in the next iteration.
The algorithm terminates when the second round is silent.
At this point, amoebot $u_r$ is the remaining active amoebot.
Feldmann et al.~\cite{FPSD21} have proven the following lemma:

\begin{lemma}
\label{lem:pasc}
    The \pasc terminates in $\lceil\log m\rceil$ iterations, resp. $O(\log m)$ rounds.
\end{lemma}

We obtain the identifiers from the \pasc as follows.
Let $k = \lceil\log m\rceil$ be the number of iterations.
For $0 \leq i < k$, let $r_i$ be the first round of the $(i + 1)$-st iteration.
Note that in each of these rounds, each amoebot of the chain receives a beep either on its primary circuit or its secondary circuit.
Hence, in round $r_i$, amoebot $x$ interprets a beep on the primary circuit as $x_i = 0$ and a beep on the secondary circuit as $x_i = 1$.

\begin{lemma}
\label{lem:id:subroutine}
    The \pasc computes $\id_{C,u_r}$.
\end{lemma}

\begin{proof}[Proof of \cref{lem:id:subroutine}]
    Since the reference amoebot $u_r$ activates its primary circuit,
    it always receives a beep on its primary circuit.
    This implies $(u_r)_i = 0$ for $0 \leq i < k$.
    Thus, Equation~\ref{eq:id_chain_1} holds.
    
    Let $y$ be the successor of $x$.
    We have to show that Equation~\ref{eq:id_chain_2} holds,
    i.e., $\id_{C,u_r}(y) = \id_{C,u_r}(x) + 1$.
    We have two cases:
    (i) $y$ never becomes passive,
    and
    (ii) $y$ becomes passive.
    The first case directly implies $\id_{C,u_r}(x) = (1, \dots, 1) = -1$ and $\id_{C,u_r}(y) = (0, \dots, 0) = 0$.
    For the second case, let $l$ denote the iteration where $y$ becomes passive.
    Hence, $y$ has received a beep on its primary circuit in the first $l - 2$ iterations and a beep on its secondary circuit in the $(l - 1)$-st iteration, i.e., $y_{l-1} = 1$ and $y_i = 0$ for $0 \leq i < l - 1$.
    Since $y$ is active, $x$ has received a beep on its secondary circuit in the first $l - 2$ iterations and a beep on its primary circuit in the $(l - 1)$-st iteration, i.e., $x_{l-1} = 0$ and $x_i = 1$ for $0 \leq i < l - 1$.
    Since $y$ is passive from the $l$-th iteration, $x_i = y_i$ holds for $l \leq i < k$.
    We obtain
    \begin{align*}
    \id_{C,u_r}(x) &= (x_{k-1}, \dots, x_l, 0, 1, \dots, 1) \\
    \id_{C,u_r}(y) &= (y_{k-1}, \dots, y_l, 1, 0, \dots, 0) = (x_{k-1}, \dots, x_l, 0, 1, \dots, 1) + 1
    \end{align*}
    since $2^{l-1} = \sum_{i=0}^{l-2} 2^i + 1$.
\end{proof}

\subsection{Spatial Identifiers}
\label{sec:idspatial}

In this section, we compute identifiers relative to the spatial positions of the amoebots with respect to a cardinal direction $d \in D_m \cup D_p$ and a reference amoebot $u_r \in S$.
Let $d'$ denote the direction obtained if we rotate $d$ by $90^\circ$ counterclockwise, e.g., $d' = N$ for $d = E$.
%
First, consider $d \in D_p$ (see Figure~\ref{fig:id_spatial_E_S}).
Let $\mathcal A_d = \{ \operatorname A(S,v,d') \mid v \in S \}$ denote a set of stripes.
We first assign identifiers to $\mathcal A_d$.
Afterwards, we extend these identifiers to the nodes in $S$.

Observe that $\mathcal A_d$ partitions $S$ into disjoint stripes.
These stripes form a chain $\mathcal C_d$ if we think of the stripes as nodes
such that two nodes are adjacent
if the corresponding stripes are neighbors (see Figure~\ref{fig:id_spatial_E_S}).
The order of the chain is given by the cardinal direction $d$:
The successor of a stripe is the neighboring stripe in direction $d$.
Let $\succr(A)$ denote the succeeding stripe of stripe $A$.
Let $A_r = \operatorname A(S,u_r,d')$.
We assign identifiers $\id_{\mathcal C_d,A_r}$ according to the following two rules:
\begin{align*}
    \id_{\mathcal C_d,A_r}(A_r) &= 0 \\ 
    \id_{\mathcal C_d,A_r}(\succr(A)) &= \id_{\mathcal C_d,A_r}(A) + 1 \text{ for all } A \in \mathcal A_d 
\end{align*}
Finally, we define $\id_{d,u_r}(v) = \id_{\mathcal C_d,A_r} (\operatorname A(S,v,d'))$ for all nodes $v \in S$.

\begin{figure}[tb]
    \begin{subfigure}[b]{0.28\textwidth}
        \centering
        \includegraphics[page=1]{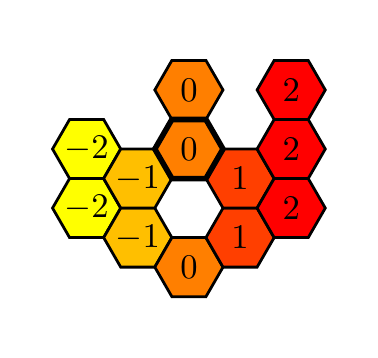}
        \caption{}
        \label{fig:id_spatial_E_S}
    \end{subfigure}
    \hfill
    \begin{subfigure}[b]{0.28\textwidth}
        \centering
        \includegraphics[page=1]{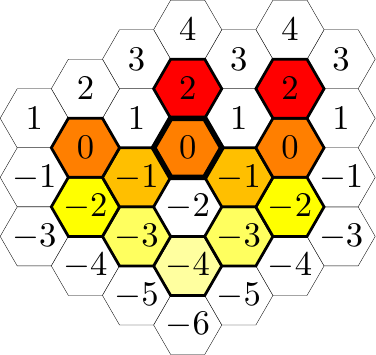}
        \caption{}
        \label{fig:id_spatial_N_S}
    \end{subfigure}
    \hfill
    \begin{subfigure}[b]{0.32\textwidth}
        \centering
        \includegraphics[page=1]{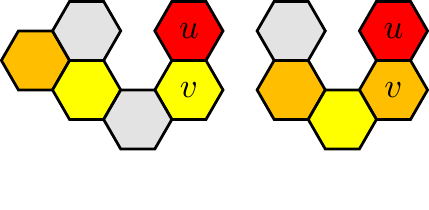}
        \caption{}
        \label{fig:id_spatial_pre}
    \end{subfigure}
    \caption{
        Spatial identifiers.
        (a) show the spatial identifiers with respect to $d = E$,
        and (b) with respect to $d = N$.
        Each color indicates a stripe (in $S$).
        The white hexagons indicate the neighborhood of $S$.
        The thick boundary indicates the reference amoebot $u_r$.
        (c) 
        If we do not include the neighborhood of $S$, then amoebot $u$ is unable to determine whether amoebot $v$ belongs to the preceding stripe (orange) or the stripe preceding its preceding stripe (yellow).
    }
    \label{fig:id_spatial}
\end{figure}

In order to compute the identifiers,
we transfer the concept of primary and secondary circuits from a chain of amoebots to a chain of stripes (compare with Section~\ref{sec:pasc}):
From a global perspective, each stripe knows its predecessor and its successor, is either active or passive, and has a primary and secondary partition set.
The primary and secondary partition sets are connected by the following rules:
If a stripe is active, its primary partition set is connected to the secondary partition set of its predecessor, and its secondary partition set is connected to the primary partition set of its predecessor.
If a stripe is passive, its primary partition set is connected to the primary partition set of its predecessor, and its secondary partition set is connected to the secondary partition set of its predecessor.
There are no further connections.
In order to keep the amoebots within a stripe synchronized, we additionally require that each amoebot has access to the primary and secondary partition set of its stripe.
In the following, we explain how to construct the circuits that satisfy the aforementioned properties.

Note that the neighborhood of each amoebot only contains amoebots of the same stripe, the preceding stripe, and the succeeding stripe.
Since we assume common compass orientation and chirality,
each amoebot is able to determine to which stripe each neighbor belongs.
We now define pin configurations that satisfy the aforementioned properties.
Each pin configuration has two partition sets that we call the primary and secondary partition set.
We connect the primary and secondary partition sets of an amoebot to the primary and secondary partition sets of adjacent amoebots
such that the aforementioned properties are reflected locally.
For example, consider two adjacent amoebots $u$ and $v$ such that $u$ belongs to an active stripe and $v$ belongs to $u$'s preceding stripe.
We connect $u$'s primary partition set to $v$'s secondary partition set,
and $u$'s secondary partition set to $v$'s primary partition set.
Since we assign the same identifiers to amoebots of the same stripe,
we connect their primary and secondary partition sets, respectively.
We obtain two pin configurations: one for amoebots of active stripes and one for amoebots of passive stripes (see Figure~\ref{fig:config}).

\begin{figure}[tb]
    \centering
    \begin{subfigure}[b]{0.16\textwidth}
        \centering
        \includegraphics[page=1]{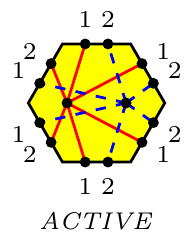}
    \end{subfigure}
    \hfill
    \begin{subfigure}[b]{0.16\textwidth}
        \centering
        \includegraphics[page=1]{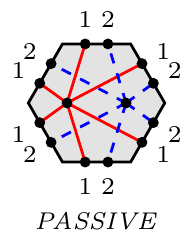}
    \end{subfigure}
    \hfill
    \begin{subfigure}[b]{0.16\textwidth}
        \centering
        \includegraphics[page=1]{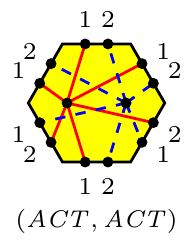}
    \end{subfigure}
    \hfill
    \begin{subfigure}[b]{0.16\textwidth}
        \centering
        \includegraphics[page=1]{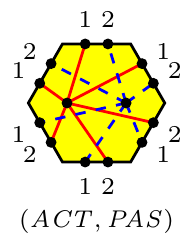}
    \end{subfigure}
    \hfill
    \begin{subfigure}[b]{0.16\textwidth}
        \centering
        \includegraphics[page=1]{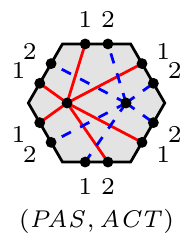}
    \end{subfigure}
    \hfill
    \begin{subfigure}[b]{0.16\textwidth}
        \centering
        \includegraphics[page=1]{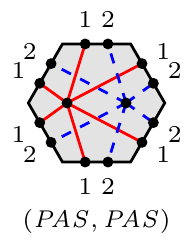}
    \end{subfigure}
    \caption{
        Utilized pin configurations.
        We utilize the former two pin configurations for $d = E$.
        We utilize the latter four pin configurations for $d = N$.
        The first argument denotes the state of the amoebot's stripe
        and the second argument denotes the state of the preceding stripe, respectively.
    }
    \label{fig:config}
\end{figure}

\begin{lemma}
\label{lem:id:construction}
    The construction satisfies the aforementioned properties.
\end{lemma}

\begin{proof}
    For the sake of analysis, we first consider an infinite amoebot structure where $S = V$.
    Afterwards, we transfer the results to arbitrary connected amoebot structures.
    
    Consider a single stripe.
    Let $\mathit{ACTIVE}$ denote the pin configuration used for amoebots of active stripes,
    and $\mathit{PASSIVE}$ denote the pin configuration used for amoebots of passive stripes.
    Both pin configurations define a primary and secondary partition set.
    All amoebots within the stripe connect their primary and secondary partition sets, respectively (see Figure~\ref{fig:psc:construction}).
    We define the union of all primary resp. secondary partition sets within a stripe as the primary resp. secondary partition set of the stripe.
    Note that each amoebot has access to both partition sets and is able to distinguish between them.
    
    Next, consider the connections to the preceding stripe.
    The connections within the pin configuration $\mathit{ACTIVE}$ are selected in such a way that an active stripe connects its primary partition set exclusively to the secondary partition set of the preceding stripe, and its secondary partition set exclusively to the primary partition set of the preceding stripe (see Figure~\ref{fig:psc:construction}).
    Similar, the connections within the pin configuration $\mathit{PASSIVE}$ are selected in such a way that an active stripe connects its primary partition set exclusively to the primary partition set of the preceding stripe, and its secondary partition set exclusively to the secondary partition set of the preceding stripe (see Figure~\ref{fig:psc:construction}).
    Note that there are no further connections.
    
    The crucial property of our construction is that any two amoebots are connected by any arbitrary path of amoebots in the same fashion,
    i.e., either both primary partition sets are connected to the secondary partition set of the other amoebot, or their primary and secondary partition sets are connected, respectively (see Figure~\ref{fig:psc:reduction:a}).
    This allows us to remove amoebots without separating the circuits as long as the amoebot structure stays connected (see Figure~\ref{fig:psc:reduction:b}).
\end{proof}

\begin{figure}[tb]
    \centering
    \begin{subfigure}[b]{0.24\textwidth}
        \centering
        \includegraphics[page=1]{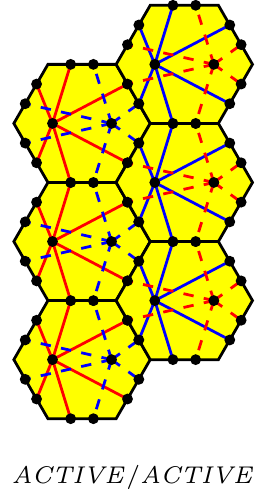}
    \end{subfigure}
    \hfill
    \begin{subfigure}[b]{0.24\textwidth}
        \centering
        \includegraphics[page=1]{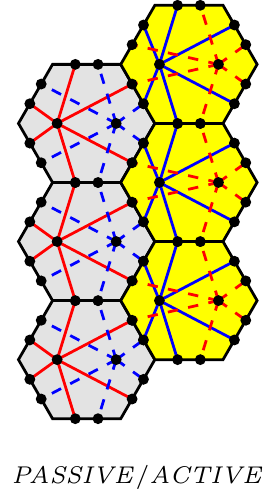}
    \end{subfigure}
    \hfill
    \begin{subfigure}[b]{0.24\textwidth}
        \centering
        \includegraphics[page=1]{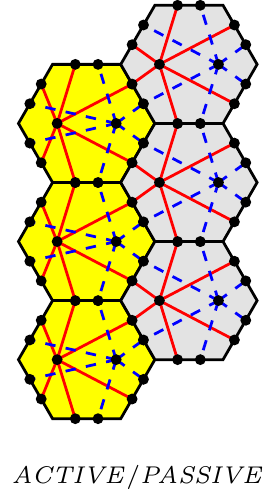}
    \end{subfigure}
    \hfill
    \begin{subfigure}[b]{0.24\textwidth}
        \centering
        \includegraphics[page=1]{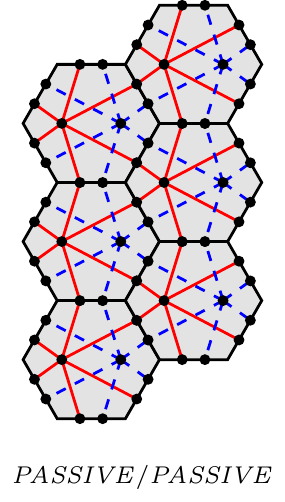}
    \end{subfigure}
    \caption{
        Sections of various stripes of an infinite amoebot structure.
    }
    \label{fig:psc:construction}
\end{figure}

\begin{figure}[tb]
    \centering
    \begin{subfigure}[b]{0.45\textwidth}
        \centering
        \includegraphics[page=1]{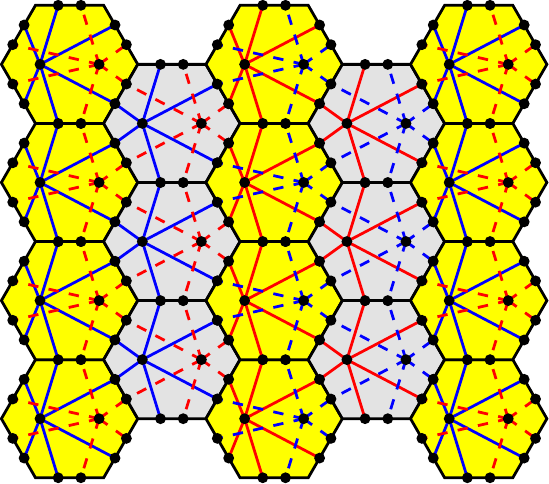}
        \caption{}
        \label{fig:psc:reduction:a}
    \end{subfigure}
    \hfill
    \begin{subfigure}[b]{0.45\textwidth}
        \centering
        \includegraphics[page=1]{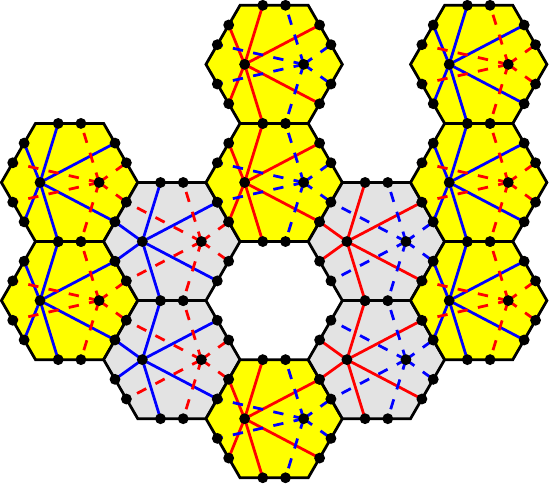}
        \caption{}
        \label{fig:psc:reduction:b}
    \end{subfigure}
    \caption{
        (a) A section of the infinite amoebot structure.
        (b) An arbitrary amoebot structure.
        The connectivity of the remaining amoebots is preserved.
    }
    \label{fig:psc:reduction}
\end{figure}

Now, we simply apply the \pasc on the chain of stripes to compute $\id_{\mathcal C_d,A_r}$, i.e., each amoebot $v \in S$ computes $\id_{\mathcal C_d,A_r} (\operatorname A(S,v,d'))$ that equals $\id_{d,u_r}(v)$ by definition.

\begin{lemma}
\label{lem:id:spatial}
    The \pasc computes $\id_{d,u_r}$ for $d \in D_p$.
\end{lemma}

\begin{proof}
    By Lemma~\ref{lem:id:construction},
    we can apply the primitive of primary and secondary circuits to the chain of stripes
    as long as each amoebot knows whether its stripe is active or passive.
    This implies that we can perform the \pasc
    that by Lemma~\ref{lem:id:subroutine}, computes $\id_{\mathcal C_d,A_r}$.
    
    It remains to show that each amoebot knows whether its stripe is active or passive throughout the execution of the algorithm.
    Initially, this is trivially true since each stripe is active.
    A stripe becomes passive once it receives a beep on its secondary circuit.
    Each amoebot of the stripe can observe this beep
    since by construction, each amoebot has access to the secondary partition set of its stripe.
    Thereafter, the stripe stays passive.
\end{proof}


%


Next, consider $d \in D_m$.
We discuss the necessary modifications in comparison to $d \in D_p$.
Note that an amoebot is unable to locally determine to which stripe its neighbors belong (see Figure~\ref{fig:id_spatial_pre}).
We therefore perform the \pasc on the union of the amoebot structure and its neighborhood (see Figure~\ref{fig:id_spatial_N_S}).
Note that each neighbor of an amoebot $v \in S$ belongs either to one of the two preceding stripes or to one of the two succeeding stripe.

Each amoebot tracks when the stripes of its neighbors become passive as follows.
Recall that all stripes are initially active.
Suppose that each amoebot $v \in S$ knows whether (its stripe and) the stripes of its neighbors are active or passive at the beginning of an iteration of the \pasc.
Hence, $v$ also knows how these stripes are interconnected.
Thus, $v$ can conclude from the signal it receives on the signals received by the stripes of its neighbors and with that whether these become passive.

There are two reasons for the tracking.
First, some of the stripes are not occupied by any amoebots.
The tracking allows each amoebot $v \in S$ to activate the correct circuits for each of its neighbors.
Second, the connections between an amoebot and its neighbor of the stripe preceding its preceding stripe depends on the states of its stripe and its preceding stripe.
We obtain four pin configurations (see Figure~\ref{fig:config}).

\begin{lemma}
\label{lem:id:spatial2}
    The \pasc computes $\id_{d,u_r}$ for $d \in D_m$.
\end{lemma}

\begin{proof}
    The proof works analogously to the one of Lemma~\ref{lem:id:spatial}.
\end{proof}

Note that we can generalize this technique to arbitrary directions as long as the identifiers within a neighborhood only differ by some constant.

\subsection{Applications}
\label{sec:idapp}

We now consider two applications for the identifiers, namely the stripe problem and the global maxima problem.

First, consider the stripe problem.
Our stripe algorithm simply executes the \pasc with $u$ as the reference amoebot, i.e., $u_r = u$.
Note that this sets the identifier of $u$ to 0, i.e., $\id_{d,u}(u) = 0$.
By construction, $\id_{d,u}(v) = \id_{d,u}(u) = 0$ holds for all $v \in \operatorname A(S,u,d)$.
We obtain the following theorem.

\begin{theorem}
\label{th:stripe}
    The stripe algorithm solves the stripe computation problem in $O(\log n)$ rounds.
\end{theorem}

Next, consider the global maxima problem.
Recall that we compute the global maxima of a set $R \subseteq S$ (see Section~\ref{sec:problem}).
By construction, 
$\argmin_{w \in R} \operatorname f_d(R,w) = \argmax_{w \in R} \id_{d,u_r}(w)$
holds for any reference amoebot $u_r$.
The idea of our global maxima algorithm is therefore to execute the \pasc and to determine the highest identifier.
By choosing an $u \in R$ as the reference amoebot, i.e., $u_r = u$, we ensure that the maximal identifier is non-negative.
In order to determine the maximum of non-negative numbers, we apply the consensus algorithm by Feldmann et al.~\cite{FPSD21} that agrees on the highest input value.
First, the amoebot structure establishes the global circuit (see Section~\ref{sec:prelim}).
Each amoebot $v \in R$ with $\id(v) \geq 0$ transmits its identifier starting from the most significant bit.
If a transmitting amoebot observes a beep in a round it does not beep, it stops its transmission.
Only an amoebot with the highest identifier is able to transmit its identifier until the end.

However, since each amoebot can only store a constant section of its identifier, and since the \pasc provides the identifiers from the least significant bit to the most significant bit, we have to partially recompute the identifiers after each bit.
In order to identify the correct bit, we simply use two binary counters
where the first counter indicates the current bit,
and the second counter the current iteration of the \pasc.
These can be realized along a chain
such that with the help of circuits, the incrementation, the decrementation, and the comparison only require $O(1)$ rounds.
Using a chain along the outer boundary set ensures a  sufficient length of the counters.
We obtain the following theorem.

\begin{theorem}
\label{th:maxima}
    The global maxima algorithm computes the global maxima in $O(\log^2 n)$ rounds w.h.p.
\end{theorem}

\section{Skeletons}

This section deals with (canonical) skeletons.
In the first subsection, we canonicalize and construct the skeletons.
In the second and third subsection, we show two applications for skeletons by showing how to construct spanning trees and how to detect symmetries.

\subsection{Canonicalized Construction}
\label{sec:skeleton}

The general idea to construct a skeleton is to start with the cycles given by the boundary sets and to fuse these into a single cycle, i.e., a skeleton
(see Figure~\ref{fig:skeleton}).
In order to fuse two boundary cycles, we have to determine a path between them.
The boundary cycles are split at the respective endpoints of the path and connected along the path.
Note that the fused cycle uses the path twice.
The difficulty lies in finding paths between the boundary cycles and in avoiding the creation of new cycles.
The construction is correct, i.e., we obtain a single cycle, iff the boundary sets and paths form a tree
with the boundary sets as the nodes and the paths as the edges.
In order to obtain a skeleton path, the skeleton is split at an arbitrary point, e.g., chosen by a leader election.


\begin{figure}[tb]
    \centering
    \begin{subfigure}[b]{0.45\textwidth}
        \centering
        \includegraphics[page=1]{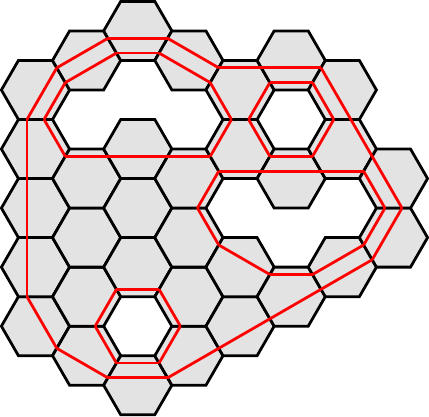}
        \caption{}
        \label{fig:skeleton_initial}
    \end{subfigure}
    \hfill
    \begin{subfigure}[b]{0.45\textwidth}
        \centering
        \includegraphics[page=1]{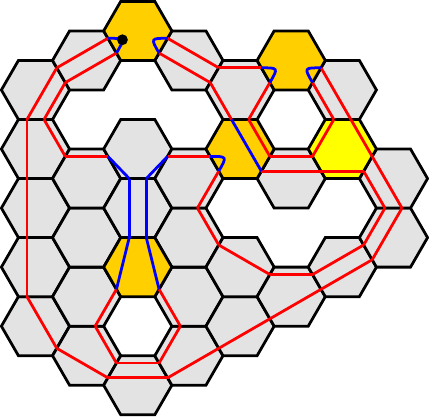}
        \caption{}
        \label{fig:skeleton_N_p}
    \end{subfigure}

    \begin{subfigure}[b]{0.45\textwidth}
        \centering
        \includegraphics[page=1]{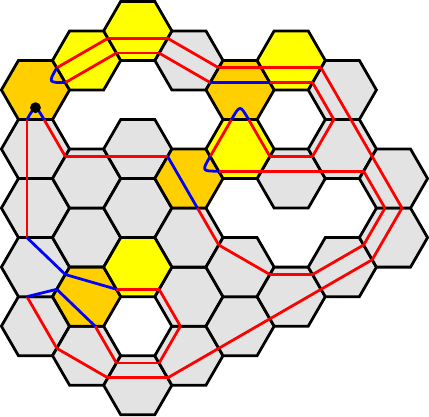}
        \caption{}
        \label{fig:skeleton_NWN_p}
    \end{subfigure}
    \hfill
    \begin{subfigure}[b]{0.45\textwidth}
        \centering
        \includegraphics[page=1]{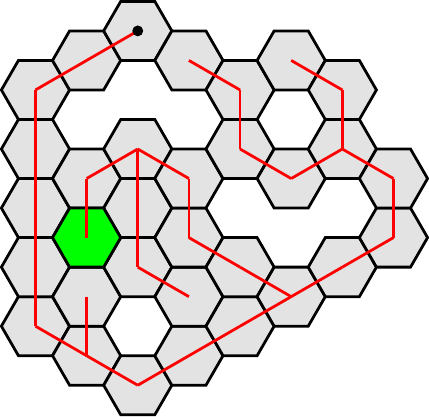}
        \caption{}
        \label{fig:spanning_tree_N_p}
    \end{subfigure}
    \caption{
        (a) shows the initial situation.
        The red lines indicate the boundary cycles.
        (b) and (c) show the $(N, +)$-skeleton and $(NNW, +)$-skeleton, respectively.
        The yellow and orange amoebots indicate the global maxima of the boundary sets.
        The orange amoebots indicate the starting points of the paths.
        The blue lines indicate the paths between the boundary cycles.
        The node indicates the location where the cycle is split.
        (d) shows the spanning tree obtained from the $(N, +)$-skeleton.
        The node indicates the root.
    }
    \label{fig:skeleton}
\end{figure}

We now canonicalize the construction of a skeleton and a skeleton path.
Recall that we define the canonical skeleton (path) with respect to a cardinal direction $d \in D_m \cup D_p$ and a sign $s \in \{ +, - \}$.
For the canonical skeleton, we have to define how the paths are constructed,
and how the boundary cycles and paths are exactly connected.
For the canonical skeleton path, we have to define the point to split the canonical skeleton.
Afterwards, we show how the amoebot structure computes the canonical skeleton in a distributed fashion.

We start with the construction of the canonical skeleton.
Let $\rho_s(d,x)$ denote the direction obtained if we rotate direction $d$ by $x$ degrees counterclockwise if the sign $s$ is positive, and clockwise if the sign $s$ is negative.
Let $d_p = d$ if $d \in D_m$ and $d_p = \rho_s(d, 30)$ if $d \in D_p$.
For each inner boundary set $B$, we construct a path as follows.
First, we compute the global maxima of $B$ with respect to direction $d$.
Let $B_d$ denote these global maxima.
Then, we compute the global maximum of $B_d$ with respect to direction $\rho_s(d,90)$.
Let $u_B$ denote the global maximum.

\begin{lemma}
\label{lem:skeleton:starting_point}
    Amoebot $u_B$ is adjacent to exactly one node in $R$, namely the one in direction $\rho_s(d_p, 180)$.
    Further, no amoebot in direction $d_p$ of $u_B$ is adjacent to a node in $R$.
\end{lemma}

\begin{proof}
    Amoebot $u_B$ has to be adjacent to at least one node in $R$.
    Otherwise, $u_B \not\in B$ would hold.
    It is easy to see that
    if that node would lie in another direction than $\rho_s(d_p, 180)$,
    either $u_B$ would not be a global maximum of $B$ with respect to $d$,
    or $u_B$ would not be the global maximum of $B_d$ with respect to direction $\rho_s(d,90)$.
    By the same reasoning, $u_B$ would not be a global maximum of $B$ with respect to $d$
    if any amoebot in direction $d_p$ of $u_B$ would be adjacent to a node in $R$.
\end{proof}

The path starts at $u_B$ and goes straight in direction $d_p$ until it reaches an amoebot $v_B$ of another boundary set.
The existence of $v_B$ is guaranteed by the outer boundary set.
Clearly, all nodes of the path are occupied by amoebots.
Note that the path may be trivial, i.e., $u_B = v_B$.
There is only a single case where $v_B$ is part of two boundary sets unequal to $B$.
In this case, we take the boundary of $v_B$ in direction $\rho_s(d_p,60)$.

\begin{lemma}
\label{lem:skeleton:tree}
    The boundary sets and paths form a tree.
\end{lemma}

\begin{proof}
    Let $\operatorname{rank}(B) = \min_{w \in B} \operatorname f_{d}(S,w)$ be the rank of an inner boundary set $B$, and let $\operatorname{rank}(B_O) = - 1$ be the rank of the outer boundary set $B_O$ (compare to the definition of the global maxima in Section~\ref{sec:problem}).
    The rank of an inner boundary set is lower than the rank of another inner boundary set if its global maxima are further in direction $d$ than the global maxima of the other inner boundary set.
    Furthermore, the outer boundary set has a lower rank than all ranks of the inner boundary sets.

    We claim that for each inner boundary set $B$, we construct a path from $B$ to another boundary set $B'$ such that $\operatorname{rank}(B) > \operatorname{rank}(B')$ holds.
    Clearly, this relationship cannot be cyclic.
    The lemma immediately follows
    since we construct a path for each inner boundary set.
    We prove the claim in the following.

    Lemma~\ref{lem:skeleton:starting_point} excludes the possibility of a self-loop, i.e., $B \neq B'$ holds.
    The claim holds by definition if $B'$ is the outer boundary set.
    Suppose that $B'$ is an inner boundary set.
    The claim also holds if the path from $u_B$ to $v_B$ is not trivial
    since $\operatorname{rank}(B) = \operatorname f_d(S,u_B) > \operatorname f_d(S,v_B) > \operatorname{rank}(B')$ holds.

    Suppose that the path is trivial, i.e., $u_B \in B$ and $u_B \in B'$.
    Let $w \in R'$ be a node adjacent to $u_B$.
    Note that $\operatorname f_d(S,u_B) \geq \operatorname f_d(S,w)$ holds
    since otherwise, $R = R'$ and with that $B = B'$ would hold.
    We go from $w$ into direction $d_p$ until we reach an amoebot $x \in B'$.
    Note that $\operatorname f_d(S,w) > \operatorname f_d(S,x)$ holds
    since for $V \setminus R_O$, $\operatorname f_d$ is strictly monotonically decreasing if we go into direction $d_p$.
    This amoebot exists since $B'$ is an inner boundary set.
    The claim holds
    since $\operatorname{rank}(B) = \operatorname f_d(S,u_B) \geq \operatorname f_d(S,w) > \operatorname f_d(S,x) > \operatorname{rank}(B')$ holds.
\end{proof}

It remains to define how the cycles and paths are exactly connected.
We define that the cycle runs along the tree without crossing itself.

Next, consider the construction of the canonical skeleton path.
We determine the splitting point $u_{B_O}$
by applying the same procedure as for the starting points of the paths
on the outer boundary set $B_O$.
That is, we first compute the global maxima of $B_O$ with respect to direction $d$,
and then compute the global maximum of these global maxima with respect to direction $\rho_s(d,90)$.
If the canonical skeleton visits $u_{B_O}$ multiple times, we pick a predefined position with respect to $d_p$ (see Figure~\ref{fig:splitting_point}).

\begin{figure}[hbt]
    \centering
    \includegraphics[page=1]{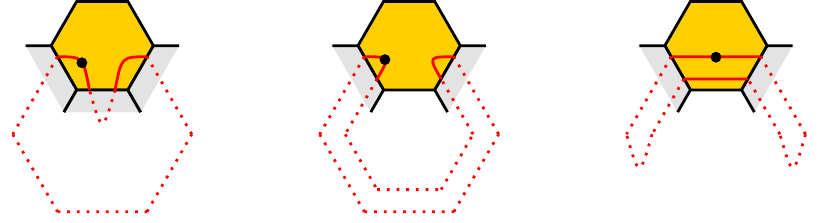}
    \caption{
        Predefined splitting point with respect to $d_p$.
        For $d_p = N$, the figure shows all cases where the canonical skeleton visits $u_{B_O}$ multiple times.
        By similar arguments as for Lemma~\ref{lem:skeleton:starting_point}, there are no other cases.
        The red lines indicate the canonical skeleton.
        The node indicates the splitting point.
    }
    \label{fig:splitting_point}
\end{figure}

\begin{lemma}
\label{lem:skeleton:complexity}
    The canonical skeleton (path) visits each bond at most twice.
    Thus, the canonical skeleton has linear complexity.
\end{lemma}

\begin{proof}
    The canonical skeleton (path) visits a bond either due to a boundary cycle or due to one of the paths.
    Each local boundary (a common unoccupied adjacent node of the endpoints) adds one visit.
    Each bond has at most two local boundaries.
    Each path adds two visits.
    Due to Lemma~\ref{lem:skeleton:starting_point}, a bond cannot be part of more than one path.
    A bond cannot be part of a boundary cycle and a path at the same time
    since the path would stop at either endpoints due to the unoccupied adjacent node.
    Altogether, each bond is visited at most twice.
\end{proof}

In the following, we present our canonical skeleton algorithm
that computes the canonical skeleton and the splitting point for the canonical skeleton path in parallel.
Some instructions are performed on different subsets in parallel.
In order to keep the amoebot structure synchronized, we apply the synchronization primitive (see Section~\ref{sec:synchronization}).
In a preprocessing step, each boundary set determines whether it is an inner or outer boundary set (see Corollary~\ref{cor:iobt}).

The canonical skeleton algorithm follows our construction of the canonical skeleton.
In the first step,
we compute the starting points of the paths and the splitting point
by performing the global maxima algorithm on each boundary set $B$ with respect to direction $d$,
and on each resulting set $B_d$ with respect to direction $\rho_s(d,90)$.
However, the computation of the boundary sets may interfere with each other since the boundary sets may intersect.
In order to circumvent that problem, we add two additional external links and let each boundary set use the two external links closer to the corresponding empty region (see Figures~\ref{fig:boundary_maxima1} to~\ref{fig:boundary_maxima3}).
Note that an amoebot is not able to determine whether two adjacent nodes belong to the same empty region (see Figure~\ref{fig:boundary_maxima3}).
Hence, it can only construct the primary and secondary circuits along the cycle.
Subsequently, it handles each of its occurrences within the cycle separately.
Nonetheless, the adjusted construction still satisfies the necessary properties given in Section~\ref{sec:idspatial}.

\begin{figure}[tb]
    \centering
    \begin{subfigure}[b]{0.25\textwidth}
        \centering
        \includegraphics[page=1]{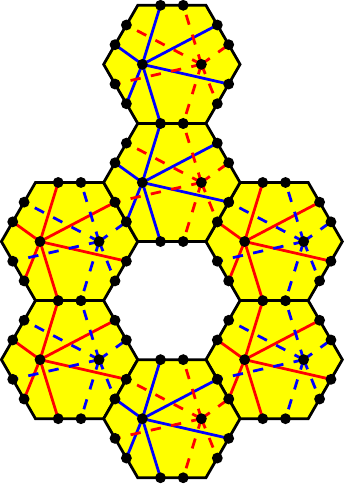}
        \caption{}
        \label{fig:boundary_maxima1}
    \end{subfigure}
    \hfill
    \begin{subfigure}[b]{0.25\textwidth}
        \centering
        \includegraphics[page=1]{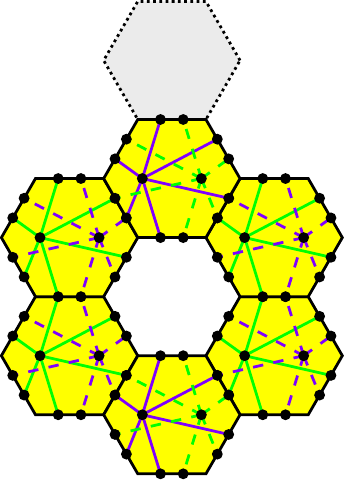}
        \caption{}
        \label{fig:boundary_maxima2}
    \end{subfigure}
    \hfill
    \begin{subfigure}[b]{0.25\textwidth}
        \centering
        \includegraphics[page=1]{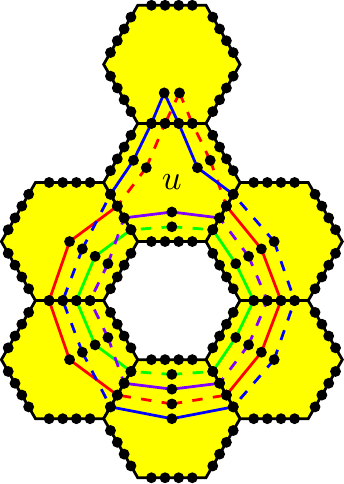}
        \caption{}
        \label{fig:boundary_maxima3}
    \end{subfigure}
    \hfill
    \begin{subfigure}[b]{0.20\textwidth}
        \centering
        \includegraphics[page=1]{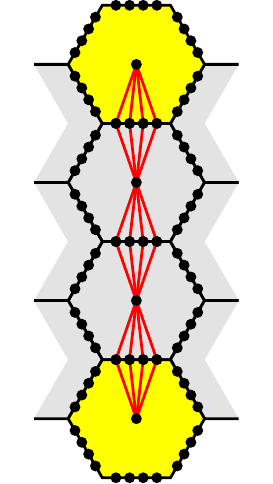}
        \caption{}
        \label{fig:boundary_maxima4}
    \end{subfigure}
    \caption{
        Computation of the global maxima of each boundary set with respect to $d = N$.
        (a) and (b) show the original construction for the outer and inner boundary set, respectively.
        The gray amoebot does not participate in the computation for the inner boundary set.
        (c) shows the construction for the computation in parallel.
        Amoebot $u$ is unaware that two of its adjacent nodes belong to the same empty region.
        (d) shows the circuit utilized to identify the paths.
        The yellow amoebots are boundary amoebots.
        The gray amoebots are inner amoebots.
    }
    \label{fig:boundary_maxima}
\end{figure}

The second step is the computation of the paths from the starting points straight into direction $d_p$.
The canonical skeleton algorithm proceeds as follows.
Each inner amoebot connects all pins of its neighbors in directions $d_p$ and $\rho_s(d_p,180)$,
and each boundary amoebot connects all pins to its neighbors in direction $d_p$ and $\rho_s(d_p,180)$, respectively (see Figure~\ref{fig:boundary_maxima4}).
Each starting point without a second boundary activates the circuit to its neighbor in direction $d_p$.
Each amoebot that receives a beep is part of a path from the starting point straight into direction $d_p$.
Finally, we obtain the following theorem:

\begin{theorem}
\label{th:skeleton}
    The canonical skeleton algorithm computes a (canonical) skeleton (path) in $O(\log^2 n)$ rounds w.h.p.
\end{theorem}

\begin{proof}
    The preprocessing step requires $O(\log n)$ rounds w.h.p. (see Section~\ref{sec:boundary}).
    The first step requires $O(\log^2 n)$ rounds for the computation of global maxima (see Section~\ref{sec:idapp}).
    The second requires $O(1)$ rounds.
    Altogether, the canonical skeleton algorithm requires $O(\log^2 n)$ rounds w.h.p.
\end{proof}

\subsection{Spanning Tree}
\label{sec:spanningtree}

We now show how a skeleton can be utilized to construct a spanning tree.
We assume that we have already computed a (not necessarily canonical) skeleton (see \cref{sec:skeleton}). 
Our spanning tree algorithm consists of two phases.
We first outline the goal of each phase.
In the first phase, we construct a tree spanning all amoebots of the skeleton path.
In the second phase, we add the remaining amoebots to the tree.
Now, consider the first phase.
We make use of the following lemma.

\begin{lemma}
\label{lem:tree}
    Let $G = (V,E)$ be a connected graph.
    Let $\pi = (v_1, \dots, v_m)$ be a path in $G$.
    Let $V' \subseteq V$ be the set of all amoebots on the path $\pi$.
    Let $\pi(v)$ denote the first edge in $\pi$ incident to $v$.
    Then, $T = (V', E')$ with $E' = \bigcup_{v \in V' \setminus \{v_1\}} \{\pi(v)\}$ is a tree.
\end{lemma}

\begin{proof}
    In order to prove that $T$ is a tree, we show that $T$ is cycle-free and connected.
    Each edge $\pi(v) = (v', v)$ implies that $v'$ appears before $v$ in $\pi$.
    Clearly, this relationship cannot be cyclic.
    
    We prove that $T$ is connected by induction on the path $\pi$.
    The induction base holds trivially for $v_1$.
    Suppose that all nodes up to node $v_i$ are connected within $T$.
    Consider node $v_{i+1}$.
    If it is not the first occurrence of $v_{i+1}$ on the path, then $v_{i+1}$ is already connected by induction hypothesis.
    Otherwise $\pi(v_{i+1}) = \{v_i, v_{i+1}\} \in E'$.
    This edge connects $v_{i+1}$ to all nodes up to node $v_i$ since these are connected by induction hypothesis.
\end{proof}

In order to determine the first occurrence of each amoebot,
we apply the \pasc algorithm on the path with $v_0$ as the reference amoebot (see \cref{sec:idchain}).
Each amoebot is able to determine its first occurrence
by simply comparing the identifiers of all its occurrences.
Each amoebot notifies the predecessor of its first occurrence.

Next, consider the second phase.
For each amoebot $v$ not included in the skeleton $S \setminus V'$, we add an edge from it to its northern neighbor $w$ to the spanning tree.
Note that $v$ is an inner amoebot
such that $w$ has to exist.
Otherwise, $v$ would be included in the skeleton.
Each amoebot $v \in S \setminus V'$ notifies its northern neighbor.
We obtain the following theorem:

\begin{theorem}
\label{th:spanning_tree}
    Given a skeleton,
    the spanning tree algorithm computes a spanning tree after $O(\log n)$ rounds.
    Altogether, it requires $O(\log^2 n)$ rounds w.h.p.
\end{theorem}

\begin{proof}
    The correctness follows from Lemma~\ref{lem:tree}.
    The first phase requires $O(\log n)$ rounds (see Section~\ref{sec:idchain}).
    The second phase requires $O(1)$ rounds.
    Altogether, the spanning tree algorithm requires $O(\log n)$ rounds.
\end{proof}

\subsection{Symmetry Detection}
\label{sec:symmetry}

We now show how to detect rotational symmetries and reflection symmetries.
Due to the underlying infinite regular triangular grid graph $\Geqt$, there is only a limited number of possible symmetries.
More precisely, an amoebot structure can only be 2-fold, 3-fold or 6-fold rotationally symmetric, and reflection symmetric to axes in a direction of $D_m \cup D_p$.
Moreover, the problem is complicated by the facts
that the symmetry point may be an unoccupied node of $\Geqt$ or not a node of $\Geqt$ at all,
and that the symmetry axis may not be occupied by any amoebots.

Recall that we define a canonical skeleton by two parameters:
the direction $d$ and the sign $s$.
Note that rotating the direction results in a rotated construction,
and inverting the sign results in a reflected construction.
Hence, a symmetric amoebot structure implies a symmetric construction of canonical skeletons.
The idea of our symmetry detection algorithm is therefore to compare the canonical skeletons.
We compare the skeletons according to the following observations (compare \cref{fig:skeletons-symmetries}).

\begin{observation}
An amoebot structure is 2-fold rotationally symmetric if the $(N,+)$-skeleton and the $(S,+)$-skeleton are symmetric.
An amoebot structure is 3-fold rotationally symmetric if the $(N,+)$-skeleton and the $(ESE,+)$-skeleton are symmetric.
An amoebot structure is 6-fold rotationally symmetric
if it is 2-fold and 3-fold rotationally symmetric.

Let $d \in D_m \cup D_p$ and let $d'$ denote the direction obtained if we rotate $d$ by $90^\circ$ counterclockwise.
An amoebot structure is reflection symmetric to an axis in direction $d$ if the $(d',+)$-skeleton and the $(d',-)$-skeleton are symmetric.
Note that due to symmetry, it is enough to only check half of $D_m \cup D_p$.
\end{observation}

In order to compare two canonical skeletons, we map each canonical $(d,s)$-skeleton path to a unique bit string by having each amoebot on the path store a partial bit string of constant length encoding the direction of its successor relative to direction $d$ and sign $s$.
Consequently, the comparison of two canonical skeletons is reduced to the comparison of the corresponding bit strings of the two skeletons.
In the following we show how such a comparison of two bit strings is possible in polylogarithmic time.


\begin{figure}[tb]
    \centering
    \begin{subfigure}[b]{0.3\textwidth}
        \centering
        \includegraphics[page=1]{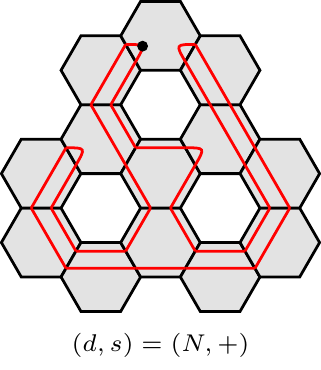}
    \end{subfigure}
    \hfill
    \begin{subfigure}[b]{0.3\textwidth}
        \centering
        \includegraphics[page=1]{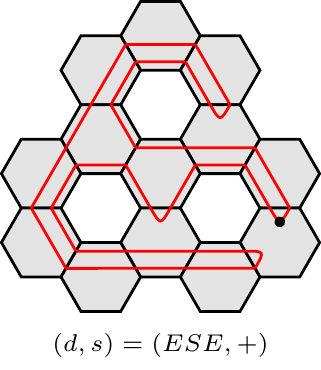}
    \end{subfigure}
    \hfill
    \begin{subfigure}[b]{0.3\textwidth}
        \centering
        \includegraphics[page=1]{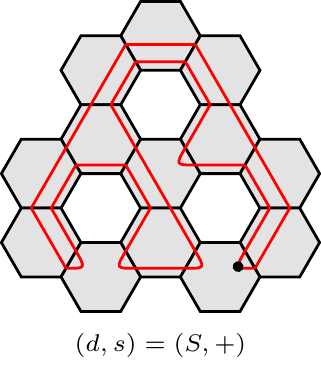}
    \end{subfigure}
    
    \begin{subfigure}[b]{0.3\textwidth}
        \centering
        \includegraphics[page=1]{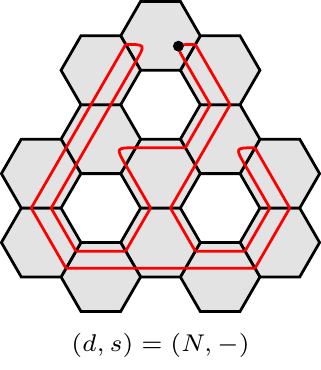}
    \end{subfigure}
    \hfill
    \begin{subfigure}[b]{0.3\textwidth}
        \centering
        \includegraphics[page=1]{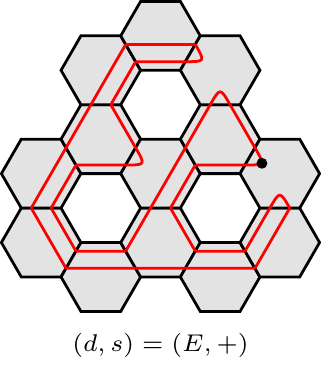}
    \end{subfigure}
    \hfill
    \begin{subfigure}[b]{0.3\textwidth}
        \centering
        \includegraphics[page=1]{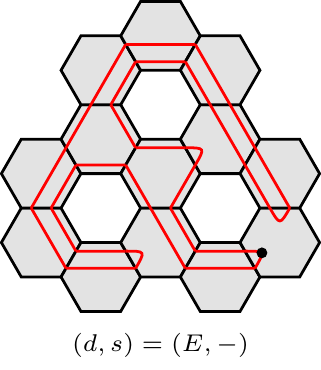}
    \end{subfigure}
    \caption{
        Symmetries of an amoebot structure.
        The amoebot structure is 3-fold rotational symmetric since the $(N,+)$- and the $(ESE,+)$-skeleton are symmetric,
        but not 2-fold or 6-fold rotational symmetric since the $(N,+)$- and the $(S,+)$-skeleton are not symmetric.
        Further, the amoebot structure is reflection symmetric to an axis into northern direction since the $(N,+)$- and the $(N,-)$-skeleton are symmetric,
        but not reflection symmetric to an axis into eastern direction since the $(E,+)$- and the $(E,-)$-skeleton are not symmetric.
    }
    \label{fig:skeletons-symmetries}
\end{figure}


To this end, we consider the \emph{string equality problem}:
Let $A=(A_0,\dotsc,A_{m-1})$ and $B=(B_0,\dotsc,B_{m'-1})$ be two chains of amoebots with reference amoebots $A_0$ and $B_0$, holding bit strings $a=(a_0,\dotsc,a_{m-1})$ and $b=(b_0,\dotsc,b_{m'-1})$. We show how $a$ and $b$ can be checked for equality in time $O(\log^5 m)$ w.h.p. using probabilistic polynomial identity testing.

We first give a high-level overview of our solution:
Since we can compare the length of $A$ and $B$ by comparing the identifiers $\id_{A,A_0}(A_{m-1})=m-1$ and $\id_{B,B_0}(B_{m'-1})=m'-1$ of the last amoebots of the chains bit by bit with the \pasc (see \cref{sec:pasc}) in time $O(\log m)$, we can assume $m=m'$ in the following.
%
Let $c\in\N$. Chain $A$ generates a prime $p\geq 2m$ and repeats the following procedure: $A$ samples $r$ uniformly at random from $[p]$ and sends the pair $(p,r)$ to chain $B$. Chain $A$ computes $f_a(r)=\sum_{i=0}^{m-1}a_i r^i\pmod p$, and chain $B$ computes $f_b(r)=\sum_{i=0}^{m-1}b_i r^i\pmod p$ and sends the result to chain $A$ which outputs ``$a\neq b$'' if $f_a(r)\neq f_b(r)$ and repeats the procedure otherwise. After $c\lceil\log m\rceil$ repetitions, $A$ outputs ``$a=b$''.
%
%
Note that $a=b$ implies $f_a(r)=f_b(r)$. From the Schwartz-Zippel lemma follows that the one-sided error probability for a single repetition is
%
$\Prob[f_a(r)=f_b(r)\mid a\neq b]\leq m/p\leq 1/2$.
It follows
$\Prob[A\text{ outputs ``}a=b\text{''}\mid a\neq b]\leq 1/m^c$.
%

We now describe the algorithm in more detail:
%
First, we describe a \emph{block primitive} that we use to divide the chain $A$ into blocks of length $k=O(\log m)$ where $k=2^{\lceil\log \lambda\rceil}$, $\lambda=2l$ and $l=\lceil\log m\rceil+2$. Note that $k\geq\lambda$. We have $k\leq m$ for $m\geq 44=:\eta$. From here on we assume $m\geq\eta$ (in case $m<\eta$, the chains $A$ and $B$ can simply compare their bit strings deterministically).
Since the \pasc terminates after $\lceil\log m\rceil$ iterations, we can easily determine amoebot $A_{\lambda}$ by using the \pasc $2$ times and forwarding a marker after every iteration. Then we use the \pasc again, with the following addition (compare \cref{fig:block-primitive}):
For an amoebot let $Q$ be the partition set on which it received a beep (either its primary or secondary partition set).
%
An active amoebot (except $A_0$) splits $Q$ into singletons.
We obtain a circuit between each pair of consecutive active amoebots.
Then, $A_0$ beeps on $Q$.
If $A_{\lambda}$ receives a beep, the procedure terminates, otherwise we continue with the next iteration.
After termination, exactly the amoebots $A_{ik}$ are active.
Since we can directly compare the bits $a_i$ of the amoebots between the last active amoebot and $A_{m-1}$ with the corresponding bits $b_i$ of chain $B$ in time $O(\log m)$, we assume in the following w.l.o.g. that $k\mid m$ holds.
This enables us to divide the amoebots of $A$ into $m/k$ chains $C_i=(A_{ik},\dotsc,A_{(i+1)k-1})$ of length $k$ with reference amoebot $A_{ik}$.

\begin{figure}[hbt]
    \centering
    \includegraphics[page=1]{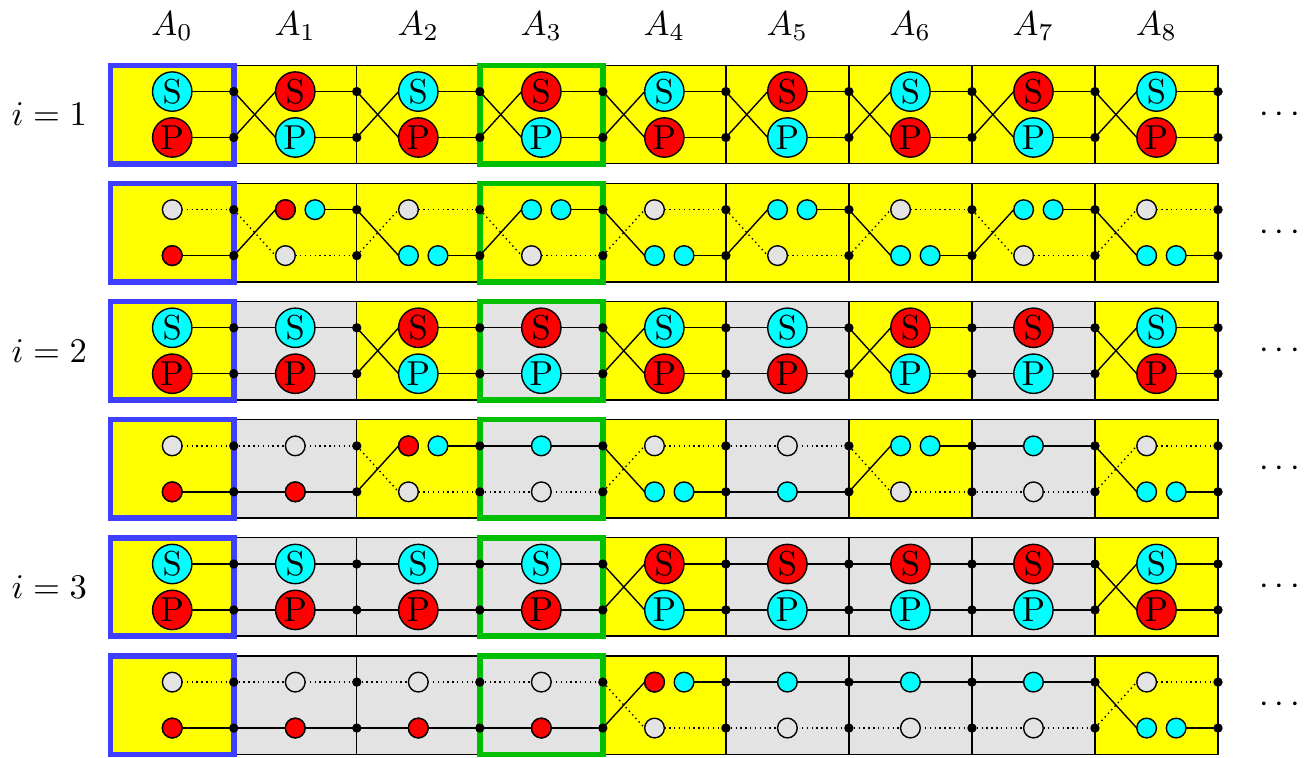}
    \caption{
        Block primitive.
        The figure shows how a chain $A$ with reference amoebot $A_0$ and a marked amoebot $A_{\lambda}$ can be divided into blocks of length $k\geq\lambda$ where $k=2^{\lceil\log\lambda\rceil}$ using the \pasc (compare \cref{fig:alg:id}) with an additional step after each iteration.
        Here, we have $\lambda=3$ and $k=4$.
        The first line shows the chain at the beginning of the $i$-th iteration of the \pasc.
        The circles indicate the partition sets. The blue bordered amoebot is $A_0$, and the green bordered amoebot is $A_{\lambda}$.
        Yellow amoebots are active, and gray amoebots are passive.
		The second line shows the configuration resulting from the following additional step:
		For an amoebot let $Q$ be the partition set on which it received a beep (depicted in red; either its primary or secondary partition set).
		An active amoebot (except $A_0$) splits $Q$ into singletons.
		We obtain a circuit between each pair of consecutive active amoebots.
		Then, $A_0$ beeps on $Q$.
		If $A_{\lambda}$ receives a beep, the procedure terminates, otherwise continue with the next iteration.
		After termination, exactly the amoebots $A_{ik}$ are active.
    }
    \label{fig:block-primitive}
\end{figure}

Now we describe how $A$ generates a prime $p\geq 2m$. Chain $A$ samples an $l$-bit integer $p=(p_0,\dotsc,p_{l-2},1)$ uniformly at random such that $A_i$ stores $p_i$. Note that the most significant bit is fixed to $1$ and therefore $p\in[2^{l-1},2^l[$, in particular $2m\leq p<4m$.
We check deterministically whether $p$ is a prime by checking in parallel for all $2\leq t<m$ whether $t\mid p$ (note that $\lfloor\sqrt{p}\rfloor<m$):
First, $p$ and $t=\id_{A,A_0}(A_{ik})=ik$ (using the \pasc) are stored in the first $l$ amoebots of every chain $C_i$ in time $O(\log m)$. Then, all chains $C_i$ repeat the following procedure in parallel for at most $k$ times: If $t\geq 2$, check whether $t\mid p$ in time $O(\log^2 m)$ using binary long division with remainder. Abort the prime testing, if $t\mid p$, otherwise increment $t$.
%

We repeat the entire procedure at most $3cl^2$ times or until we have successfully sampled a prime $p$. The runtime for the prime generation is $O(\log^5 m)$.
We now analyse the probability for the event $E_{\text{fail}}$ that no prime is generated. Using non-asymptotic bounds on the prime-counting function, one can show that the fraction of integers in $[2^{l-1},2^l[$ that are prime is at least $1/(3l)$. It follows:
$\Prob[E_{\text{fail}}]\leq (1-1/(3l))^{3cl^2}\leq 1/e^{cl}\leq 1/m^c$
%

We now address the probabilistic polynomial identity testing, focusing on the computation of $f_a(r)$ for $r\in[p]$. 
We use the previously determined division of the chain $A$ into blocks of length $k=O(\log m)$. Assume that $p$, $r$ and $e=\id_{A,A_0}(A_{ik})=ik$ are stored in the first $l$ amoebots of every chain $C_i$.
All chains $C_i$ repeat the following procedure in parallel for $k$ times: Compute $s^{(i)}=a_e r^e\pmod p$ using modular exponentiation via the right-to-left binary method. Using binary long multiplication and division with remainder, this step is possible in time $O(\log^3 m)$. Then, increment $e$.
Once all chains $C_i$ have completed the $j$-th repetition, we compute the sum of the $s^{(i)}$ modulo $p$ and store it in the first $l$ amoebots of chain $A$ using a generalization of \cref{th:addtree}. The summation is possible in time $O(\log^2 m)$. The computed sum is then added to a running total modulo $p$.

Finally, after $k$ repetitions, the result $f_a(r)$ is stored in the amoebots $A_0,\dotsc,A_{l-1}$.
The runtime for the polynomial identity testing is $O(\log^4 m)$.

Note that the size of the outer boundary set is $\Omega(\sqrt{n})$, which is also a lower bound for the size of a canonical skeleton. Hence, we get the following result:

\begin{theorem}
	\label{th:symmetry}
	The string equality problem	on chains of length $O(m)$ can be solved in $O(\log^5 m)$ rounds w.h.p. Therefore, the symmetry detection problem can be solved in $O(\log^5 n)$ rounds w.h.p.
\end{theorem}


Additionally, we can compute the amoebot occupying the symmetry point and amoebots on the symmetry axis,
but due to the similarity to the applications in \cref{sec:idspatial},
we just sketch the algorithm.
The idea is to identify some symmetric amoebots, e.g., by computing global maxima,
and to compute symmetric identifiers with these as reference amoebots (see \cref{fig:symmetry_axis_point}).
We output all amoebots that receive the same identifier for each reference amoebot.

\begin{figure}[hbt]
    \centering
    \begin{subfigure}[b]{0.2\textwidth}
        \centering
        \includegraphics[page=1]{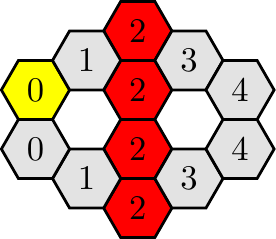}
        \caption{}
    \end{subfigure}
    \hfill
    \begin{subfigure}[b]{0.2\textwidth}
        \centering
        \includegraphics[page=1]{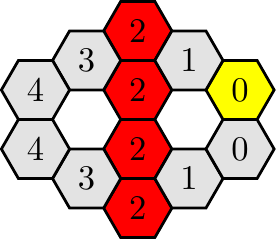}
        \caption{}
    \end{subfigure}
    \hfill
    \begin{subfigure}[b]{0.17\textwidth}
        \centering
        \includegraphics[page=1]{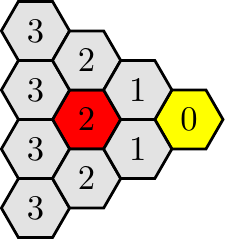}
        \caption{}
    \end{subfigure}
    \hfill
    \begin{subfigure}[b]{0.17\textwidth}
        \centering
        \includegraphics[page=1]{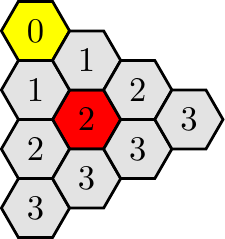}
        \caption{}
    \end{subfigure}
    \hfill
    \begin{subfigure}[b]{0.17\textwidth}
        \centering
        \includegraphics[page=1]{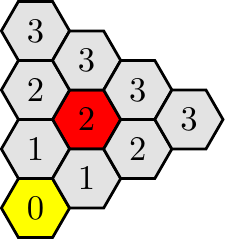}
        \caption{}
    \end{subfigure}
    
    ~
    
    \begin{subfigure}[b]{0.2\textwidth}
        \centering
        \includegraphics[page=1]{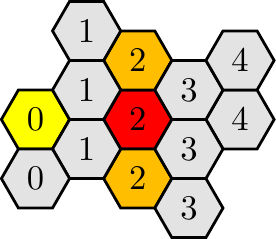}
        \caption{}
    \end{subfigure}
    \hfill
    \begin{subfigure}[b]{0.2\textwidth}
        \centering
        \includegraphics[page=1]{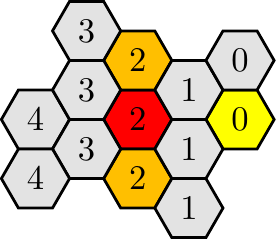}
        \caption{}
    \end{subfigure}
    \hfill
    \begin{subfigure}[b]{0.2\textwidth}
        \centering
        \includegraphics[page=1]{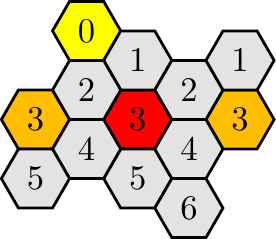}
        \caption{}
    \end{subfigure}
    \hfill
    \begin{subfigure}[b]{0.2\textwidth}
        \centering
        \includegraphics[page=1]{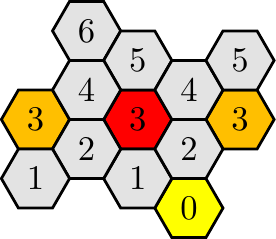}
        \caption{}
    \end{subfigure}
    \caption{
        Symmetry point and symmetry axis.
        (a) and (b) show the computation of a symmetry axis.
        (c) to (e) show the computation of a symmetry point for a 3-fold or 6-fold rotationally symmetric amoebot structure.
        (f) to (i) show the computation of a symmetry point for a 2-fold rotationally symmetric amoebot structure.
        In this case, we have to compute two axes so that the symmetry point lies on the intersection of these.
        The yellow amoebot indicates the reference amoebot.
        The red (and orange) amoebots receive the same identifier for each reference amoebot.
    }
    \label{fig:symmetry_axis_point}
\end{figure}

\section{Conclusion and Future Work}

In this paper, we have proposed polylogarithmic-time solutions for a range of problems.
First, we have computed spatial identifiers in order to compute a stripe through a given amoebot and direction, and the global maxima of the given amoebot structure with respect to a direction.
Using these results, we have constructed a canonical skeleton path, which provides a unique characterization of the shape of the given amoebot structure.
Constructing canonical skeleton paths for different directions will then allow the amoebots to set up a spanning tree and to check symmetry properties of the given amoebot structure.

Our solutions could be useful for various applications  like rapid shape transformation, energy dissemination, and structural monitoring.
The details have to be worked out in future work.
Beyond that, we think that exploring further applications for the spatial identifiers and the skeleton would be interesting.



%


\bibliography{literature}

\clearpage



\end{document}